\documentclass[12pt]{article}
\usepackage{geometry,mathtools}
\usepackage{amsmath,amssymb,amsthm, amsfonts}
\usepackage{etaremune, enumerate, float, verbatim}
\usepackage{algorithm}
\usepackage{algorithmic}
\usepackage{hyperref}
\usepackage{graphicx}
\usepackage{url}
\usepackage[mathlines]{lineno}
\usepackage{dsfont} % needed for double-struck 1
\usepackage{tikz, graphicx, subcaption, caption}% added CR 12/6/17
\usepackage[algo2e,ruled,vlined]{algorithm2e}
\usepackage{mathrsfs,amsmath}
\usepackage{color}
\definecolor{red}{rgb}{1,0,0}

\definecolor{blue}{rgb}{0,0,.7}

\definecolor{green}{rgb}{0,.6,0}

\definecolor{purp}{rgb}{.5,0,.5}

\numberwithin{figure}{section}   % added LH 11/15/17

\newtheorem{thm}{Theorem}[section]

\newtheorem{lem}[thm]{Lemma}

\theoremstyle{definition}

\theoremstyle{definition}

\theoremstyle{definition}

\newcommand{\cov}{\operatorname{Cov}}
\newcommand{\var}{\operatorname{Var}}
\newcommand{\opts}{\operatorname{opt_{std}}}
\newcommand{\optb}{\operatorname{opt_{bandit}}}

\newcommand{\e}{\operatorname{\mathbf{E}}}

\newcommand{\rref}{\operatorname{RREF}}

\newcommand{\bit}{\begin{itemize}}
\newcommand{\eit}{\end{itemize}}
\newcommand{\ben}{\begin{enumerate}}
\newcommand{\een}{\end{enumerate}}
\newcommand{\beq}{\begin{equation}}
\newcommand{\eeq}{\end{equation}}
\newcommand{\bea}{\begin{eqnarray*}} % * means no number
\newcommand{\eea}{\end{eqnarray*}}
\newcommand{\bpf}{\begin{proof}}
\newcommand{\epf}{\end{proof}\ms}
\newcommand{\bmt}{\begin{bmatrix}}
\newcommand{\emt}{\end{bmatrix}}
\newcommand{\ms}{\medskip}

\newcommand{\noi}{\noindent}

\title{A note on the price of bandit feedback for mistake-bounded online learning}
\author{Jesse Geneson}

\begin{document}
%\linenumbers
\maketitle

\begin{abstract}
The standard model and the bandit model are two generalizations of the mistake-bound model to online multiclass classification. In both models the learner guesses a classification in each round, but in the standard model the learner recieves the correct classification after each guess, while in the bandit model the learner is only told whether or not their guess is correct in each round. For any set $F$ of multiclass classifiers, define $\opts(F)$ and $\optb(F)$ to be the optimal worst-case number of prediction mistakes in the standard and bandit models respectively. 

Long (Theoretical Computer Science, 2020) claimed that for all $M > 2$ and infinitely many $k$, there exists a set $F$ of functions from a set $X$ to a set $Y$ of size $k$ such that $\opts(F) = M$ and $\optb(F) \ge (1 - o(1))(|Y|\ln{|Y|})\opts(F)$. The proof of this result depended on the following lemma, which is false e.g. for all prime $p \ge 5$, $s = \mathbf{1}$ (the all $1$ vector), $t = \mathbf{2}$ (the all $2$ vector), and all $z$. 

Lemma: Fix $n \ge 2$ and prime $p$, and let $u$ be chosen uniformly at random from $\left\{0, \dots, p-1\right\}^n$. For any $s, t \in \left\{1, \dots, p-1\right\}^n$ with $s \neq t$ and for any $z \in \left\{0, \dots, p-1\right\}$, we have $\Pr(t \cdot u = z \mod p \text{     } | \text{     } s \cdot u = z \mod p) = \frac{1}{p}$. 

We show that this lemma is false precisely when $s$ and $t$ are multiples of each other mod $p$. Then using a new lemma, we fix Long's proof.

\end{abstract}

\noi {\bf Keywords}: online learning, bandit feedback, mistake-bound model, learning theory

%%%%%%%%%%%%%%%%%%%%%%%%%%%%%%%%%%%%%%%%%%%
\section{Introduction}

Auer et al. \cite{auer1} introduced two generalizations of the mistake-bound model \cite{ls} called the \emph{standard model} and the \emph{bandit model}. Let $F$ be a set of functions from some set $X$ to a finite set $Y$. In the standard model, the adversary selects some $f \in F$ that the learner does not know. In each round $t$ of the learning process, the adversary gives the learner some $x_t \in X$, the learner predicts the output of $f$ with input $x_t$, and the adversary tells them the correct value of $f(x_t)$. The bandit model is similar, except at the end of each round, the adversary tells the learner \emph{yes} or \emph{no} instead of the correct value of $f(x_t)$. In both models, in any round the adversary may change the function $f$ to any other function in $F$ as long as it is consistent with their previous answers.

The goal of the learner in each model is to minimize the number of prediction mistakes, while the adversary wants to maximize the number of prediction mistakes. Define $\opts(F)$ and $\optb(F)$ to be the number of prediction mistakes in the standard and bandit models respectively if both the learner and adversary play optimally. 

Long \cite{long} proved that $\optb(F) \le (1 + o(1))(|Y|\ln{|Y|})\opts(F)$ for all such $F$, and claimed that the upper bound is best possible up to the leading constant. In order to show that the upper bound is best possible up to the leading constant, Long claimed that for all $M > 2$ and infinitely many $k$, there exists a set $F$ of functions from a set $X$ to a set $Y$ of size $k$ such that $\opts(F) = M$ and $\optb(F) \ge (1 - o(1))(|Y|\ln{|Y|})\opts(F)$. The proof used probabilistic methods inspired by \cite{rao1, rao2, cw, lw}. 

In particular, part of the proof used Chebyshev's inequality and required a set of random variables to be pairwise independent. The pairwise independence was proved using the following lemma, which is false e.g. for all prime $p \ge 5$, $s = \mathbf{1}$ (the all $1$ vector), $t = \mathbf{2}$ (the all $2$ vector), and all $z$. The error in the proof of the following lemma occurs on the last ``='' in Appendix B of \cite{long}.

\begin{lem} \label{falselem} \cite{long}
Fix $n \ge 2$ and prime $p$, and let $u$ be chosen uniformly at random from $\left\{0, \dots, p-1\right\}^n$. For any $s, t \in \left\{1, \dots, p-1\right\}^n$ with $s \neq t$ and for any $z \in \left\{0, \dots, p-1\right\}$, we have $\Pr(t \cdot u = z \mod p \text{     } | \text{     } s \cdot u = z \mod p) = \frac{1}{p}$. 
\end{lem}

Neither $s$ nor $t$ is the all $0$ vector, so $s$ is a multiple of $t$ mod $p$ if and only if $t$ is a multiple of $s$ mod $p$. We can see that Lemma \ref{falselem} is false when $s$ and $t$ are multiples of each other mod $p$. 

If $z = 0$ and $s$ and $t$ are multiples of each other mod $p$, then $\Pr(t \cdot u = z \mod p \text{     } | \text{     } s \cdot u = z \mod p) = 1$. On the other hand if $z \neq 0$ and $s$ and $t$ are multiples of each other mod $p$ with $s \neq t$, then $\Pr(t \cdot u = z \mod p \text{     } | \text{     } s \cdot u = z \mod p) = 0$. In the next section, we show that Lemma \ref{falselem} is true when $s$ and $t$ are not multiples of each other mod $p$. We use this fact to fix the proof from \cite{long} and show that for all $M > 2$ and infinitely many $k$, there exists a set $F$ of functions from a set $X$ to a set $Y$ of size $k$ such that $\opts(F) = M$ and $\optb(F) \ge (1 - o(1))(|Y|\ln{|Y|})\opts(F)$.

\section{New proof}

In the proof of the main result, we will use the following lemma from \cite{long}.

\begin{lem}\label{long1}
\cite{long} Fix $n \ge 1$, and let $u$ be chosen uniformly at random from $\left\{0, \dots, p-1\right\}^n$. For any $x \in \left\{0, \dots, p-1\right\}^n - \left\{\mathbf{0}\right\}$ and for any $y \in \left\{0, \dots, p-1\right\}$, $P(x \cdot u = y \mod p) = \frac{1}{p}$. 
\end{lem}

Now we prove a new lemma. We will use this in place of the false Lemma \ref{falselem} from \cite{long} when we prove that for all $M > 2$ and infinitely many $k$, there exists a set $F$ of functions from a set $X$ to a set $Y$ of size $k$ such that $\opts(F) = M$ and $\optb(F) \ge (1 - o(1))(|Y|\ln{|Y|})\opts(F)$.

\begin{lem} \label{newlem}
Fix $n \ge 2$, and let $u$ be chosen uniformly at random from $\left\{0, \dots, p-1\right\}^n$. For any $s, t \in \left\{1, \dots, p-1\right\}^n$ \emph{that are not multiples of each other mod $p$} and for any $z \in \left\{0, \dots, p-1\right\}$, we have $\Pr(t \cdot u = z \mod p \text{     } | \text{     } s \cdot u = z \mod p) = \frac{1}{p}$. 
\end{lem}

\begin{proof}
By Lemma \ref{long1} and the definition of conditional probability, we have $\Pr(t \cdot u = z \mod p \text{     } | \text{     } s \cdot u = z \mod p) = \frac{\Pr(t \cdot u = z \mod p \text{     } \wedge \text{     } s \cdot u = z \mod p)}{\Pr( s \cdot u = z \mod p)} = p \Pr(t \cdot u = z \mod p \text{     } \wedge \text{     } s \cdot u = z \mod p)$. Moreover $\Pr(t \cdot u = z \mod p \text{     } \wedge \text{     } s \cdot u = z \mod p) = \frac{|\left\{u: \text{     } t \cdot u = z \mod p \text{     } \wedge \text{     } s \cdot u = z \mod p \text{     } \wedge \text{     } u \in \left\{0, \dots, p-1 \right\}^n \right\}|}{p^n}$.

In order to calculate $|\left\{u:\text{     } t \cdot u = z \mod p \text{     } \wedge \text{     } s \cdot u = z \mod p \text{     } \wedge \text{     } u \in \left\{0, \dots, p-1 \right\}^n \right\}|$, we must find the number of solutions $u \in \left\{0, \dots, p-1 \right\}^n$ to the system of equations $t \cdot u = z \mod p$ and $s \cdot u = z \mod p$.

Treating $s$ and $t$ as row vectors, we form the augmented matrix $\begin{bmatrix} s & z\\ t & z \end{bmatrix}$ and row-reduce it. Since $s$ and $t$ are not multiples of each other mod $p$, they are therefore linearly independent, so $\rref(\begin{bmatrix} s & z\\ t & z \end{bmatrix})$ has two pivot entries. Therefore the system of equations $t \cdot u = z \mod p$ and $s \cdot u = z \mod p$ has two dependent variables $u_i$ and $u_j$ for some $i \neq j$ and $n-2$ independent variables $u_k$ with $k \neq i$ and $k \neq j$. There are $p$ choices for each of the independent variables, and the dependent variables are determined by the values of the independent variables, so there are $p^{n-2}$ solutions $u \in \left\{0, \dots, p-1 \right\}^n$ to the system of equations  $t \cdot u = z \mod p$ and $s \cdot u = z \mod p$.

Thus $\Pr(t \cdot u = z \mod p \text{     } \wedge \text{     } s \cdot u = z \mod p) = \frac{|\left\{u: \text{     } t \cdot u = z \mod p \text{     } \wedge \text{     } s \cdot u = z \mod p \text{     } \wedge \text{     } u \in \left\{0, \dots, p-1 \right\}^n \right\}|}{p^n} = \frac{p^{n-2}}{p^n} = \frac{1}{p^2}$, so $\Pr(t \cdot u = z \mod p \text{     } | \text{     } s \cdot u = z \mod p) = p \Pr(t \cdot u = z \mod p \text{     } \wedge \text{     } s \cdot u = z \mod p) = \frac{1}{p}$.
\end{proof}

%In the last lemma where the proof would break down if $s$ and $t$ were not linearly independent. Specifically, if $s$ and $t$ are multiples of each other mod $p$ then $\rref(\begin{bmatrix} s & z\\ t & z \end{bmatrix})$ has one pivot entry. If $z = 0$, the row of $\rref(\begin{bmatrix} s & z\\ t & z \end{bmatrix})$ with no pivot entry has a $0$ in the rightmost column, so the system of equations $t \cdot u = z \mod p$ and $s \cdot u = z \mod p$ has one dependent variables $u_i$ for some $i$ and $n-1$ independent variables $u_k$ with $k \neq i$. So there are $p^{n-1}$ solutions to the system of equations  $t \cdot u = z \mod p$ and $s \cdot u = z \mod p$, implying $\Pr(t \cdot u = z \mod p \text{     } | \text{     } s \cdot u = z \mod p) = p \frac{p^{n-1}}{p^n} = 1$. If $z \neq 0$, the row of $\rref(\begin{bmatrix} s & z\\ t & z \end{bmatrix})$ with no pivot entry has a nonzero entry in the rightmost column, so the system of equations $t \cdot u = z \mod p$ and $s \cdot u = z \mod p$ has no solutions, implying $\Pr(t \cdot u = z \mod p \text{     } | \text{     } s \cdot u = z \mod p) = 0$. 

With this new lemma, we obtain the following lemma which is analogous to a lemma in \cite{long} that followed from the false Lemma \ref{falselem}.

\begin{lem}  \label{mainlem}
For any subset $S \subset \left\{1, \dots, p-1\right\}^n$, there is an element $u \in \left\{0, \dots, p-1\right\}^n$ such that for all $z \in \left\{0, \dots, p-1\right\}$, $|\left\{x \in S: x \cdot u = z \mod p \right\}| \le \frac{|S|}{p}+2\sqrt{|S|}$.
\end{lem}

\begin{proof}
Suppose that $S$ is any subset of $\left\{1, \dots, p-1\right\}^n$, and let $u$ be chosen uniformly at random from $\left\{0, \dots, p-1\right\}^n$. For each $z \in \left\{0, \dots, p-1\right\}$, define $T_z$ as the set of $x \in S$ for which $x \cdot u = z \mod p$. By Lemma \ref{long1} and linearity of expectation, we have $\e(|T_z|) = \frac{|S|}{p}$ for all $z$. By Lemma \ref{long1}, Lemma \ref{newlem}, and the definition of $S$, the events $s \cdot u = z$ mod $p$ and $t \cdot u = z$ mod $p$ are pairwise independent for any distinct $s, t \in S$ that are not multiples of each other mod $p$. We split into two cases for $z$, $z \neq 0$ and $z = 0$. 

First, suppose that $z \neq 0$. For each $s \in S$, define the indicator random variable $X_{s, z}$ so that $X_{s, z} = 1$ if $s \cdot u = z \mod p$, and $X_{s, z} = 0$ otherwise. If $s$ and $t$ are not multiples of each other mod $p$, then $\cov(X_{s, z}, X_{t, z}) = 0$. If $s$ and $t$ are multiples of each other mod $p$ with $s \neq t$, then $\cov(X_{s, z}, X_{t, z}) = \e(X_{s, z} X_{t, z}) - \e(X_{s,z}) \e(X_{t,z}) = -\frac{1}{p^2} < 0$. Since $|T_z| = \sum_{s \in S} X_{s, z}$, we have $\var(|T_z|) = \var(\sum_{s \in S} X_{s, z}) = \sum_{s \in S} \var(X_{s, z}) + \sum_{s \neq t} \cov(X_{s, z}, X_{t, z}) \le \sum_{s \in S} \var(X_{s, z}) = |S|(\frac{1}{p}-\frac{1}{p^2}) < \frac{|S|}{p}$. By Chebyshev's inequality, $\Pr(|T_z| \ge \frac{|S|}{p}+2\sqrt{|S|})  \le \frac{1}{4p}$.

Now, suppose that $z = 0$. For each $s \in S$, define the indicator random variable $X_{s, z}$ so that $X_{s, z} = 1$ if $s \cdot u = z \mod p$, and $X_{s, z} = 0$ otherwise. If $s$ and $t$ are not multiples of each other mod $p$, then $\cov(X_{s, z}, X_{t, z}) = 0$. If $s$ and $t$ are multiples of each other mod $p$ with $s \neq t$, then $\cov(X_{s, z}, X_{t, z}) = \e(X_{s, z} X_{t, z}) - \e(X_{s,z}) \e(X_{t,z}) = \frac{1}{p} -\frac{1}{p^2} < \frac{1}{p}$. Note that there are at most $(p-2)|S|$ ordered pairs $(s, t)$ for which $s$ and $t$ are multiples of each other mod $p$ with $s \neq t$. Since $|T_z| = \sum_{s \in S} X_{s, z}$, we have $\var(|T_z|) = \var(\sum_{s \in S} X_{s, z}) = \sum_{s \in S} \var(X_{s, z}) + \sum_{s \neq t} \cov(X_{s, z}, X_{t, z}) < \frac{|S|}{p}+\frac{(p-2)|S|}{p} < |S|$. By Chebyshev's inequality, $\Pr(|T_z| \ge \frac{|S|}{p}+2\sqrt{|S|}) \le \frac{1}{4}$.

By the union bound, $\Pr(\forall z |T_z| \le \frac{|S|}{p}+2\sqrt{|S|}) \ge 1-(p-1)\frac{1}{4p}-\frac{1}{4} \ge \frac{1}{2}$. Thus we can choose $u$ randomly, and with probability at least $\frac{1}{2}$ we will have $|\left\{x \in S: x \cdot u = z \mod p \right\}| \le \frac{|S|}{p}+2\sqrt{|S|}$ for all $z \in \left\{0, \dots, p-1\right\}$.
\end{proof}

The proof of the following theorem is the same as in \cite{long}, we include it for completeness. Let $p$ be any prime number. For all $a \in \left\{0, \dots, p-1\right\}^n$, define $f_a: \left\{0, \dots, p-1\right\}^n \rightarrow \left\{0, \dots, p-1\right\}$ so that $f_a(x) = a \cdot x \mod p$ and define $F_L(p, n) = \left\{f_a: a \in \left\{0, \dots, p-1\right\}^n \right\}$. It is known that $\opts(F_L(p, n)) = n$ for all primes $p$ and $n > 0$ \cite{shv, auer1, bl, long}.

\begin{thm}
For all $M > 2$ and infinitely many $k$, there exists a set $F$ of functions from a set $X$ to a set $Y$ of size $k$ such that $\opts(F) = M$ and $\optb(F) \ge (1 - o(1))(|Y|\ln{|Y|})\opts(F)$.
\end{thm}

\begin{proof}
Fix $n \ge 3$ and prime $p \ge 5$. We let $F = F_L(p, n)$, with $X = \left\{0, \dots, p-1\right\}^n$ and $Y = \left\{0, \dots, p-1\right\}$. Let $S = \left\{1, \dots, p-1\right\}^n$, so $|S| = (p-1)^{n}$.

Let $R_1 = \left\{f_a: a \in S\right\} \subset F_L(p, n)$. In each round $t > 1$ the adversary creates a list $R_t$ of members of $\left\{f_a: a \in S\right\}$ that are consistent with its previous answers, it always answers \emph{no}, and it picks $x_t$ for round $t$ that minimizes $\displaystyle \max_{\hat{y}_t} |R_t \cap \left\{f: f(x_t) = \hat{y}_t \right\}|$. 

By Lemma \ref{mainlem}, we have $|R_{t+1}| \ge |R_t| - \frac{|R_t|}{p}-2 \sqrt{|R_t|}  \ge |R_t| - \frac{|R_t|}{p}-\frac{2|R_t|}{p \sqrt{\ln{p}}} = (1-\frac{1+\frac{2}{\sqrt{\ln{p}}}}{p})|R_t|$, as long as $|R_t| \ge p^2 \ln{p}$.

By induction on the previous inequality, we have $|R_t| \ge  (1-\frac{1+\frac{2}{\sqrt{\ln{p}}}}{p})^{t-1}(p-1)^{n}$. If $(1-\frac{1+\frac{2}{\sqrt{\ln{p}}}}{p})^{b-1}(p-1)^{n} \ge p^2 \ln{p}$, then the adversary can guarantee $b$ wrong guesses before $|R_t| < p^2 \ln{p}$. This is true for $b = (1-o(1)) n p \ln{p}$, completing the proof.
\end{proof}
 
In the last proof, we assumed that $M \ge 3$, which is fine because of the $o(1)$ in the bound. One of the open problems from \cite{long} was to determine whether the $\Omega(k \ln{k})$ lower bound still holds for $M = 2$. It is false for $M = 1$, since $\optb(F) \le k-1$ if $\opts(F) = 1$ \cite{long}.

%%%%%%%%

\end{document}